\documentclass[pra,aps,twocolumn,superscriptaddress,showpacs]{revtex4}
\usepackage{graphicx,graphics,epsfig}
\usepackage{dcolumn}
\usepackage{bm}
\usepackage{amsmath}
\usepackage{verbatim}
\usepackage{color}
\usepackage[colorlinks=false]{hyperref} 
\usepackage{subfigure}
\usepackage{times,natbib}
\usepackage{amsmath,amsfonts,amssymb,graphics,graphicx,epsfig,color,times,natbib}
\usepackage{tikz}
\usepackage{verbatim}
\usepackage[standard]{ntheorem}

\newcommand{\tr}{\mbox{Tr}}

\date{\today}

\begin{document}

\title{On realizing Lov\'asz-optimum orthogonal representation in the real Hilbert space}

\author{Zhen-Peng Xu}
 \affiliation{Theoretical Physics Division, Chern Institute of Mathematics, Nankai University,
 Tianjin 300071, People's Republic of China}

\author{Jing-Ling Chen}
\email{chenjl@nankai.edu.cn}
 \affiliation{Theoretical Physics Division, Chern Institute of Mathematics, Nankai University,
 Tianjin 300071, People's Republic of China}
 \affiliation{Centre for Quantum Technologies, National University of Singapore,
 3 Science Drive 2, Singapore 117543}

\begin{abstract}
Quantum contextuality is usually revealed by the non-contextual inequality, which can always be associated with an exclusivity graph. The quantum upper bound of the inequality is nothing but the Lov\'asz number of the graph. In this work, we show that if there is a Lov\'asz-optimum orthogonal representation realized in the $d$-dimensional complex Hilbert space, then there always exists a corresponding Lov\'asz-optimum orthogonal representation in the $(2d-1)$-dimensional real Hilbert space. This in turn completes the proof that the Lov\'asz-optimum orthogonal representation for any exclusivity graph can always be realized in the real Hilbert space of suitable dimension.
\end{abstract}


\maketitle

\section{Introduction}
Quantum contextuality, which is a fundamental concept in quantum information theory, was independently discovered by Bell \cite{bell66}, Kochen and Specker (KS) \cite{KS}. Contextuality is usually revealed by the non-contextual inequality, quantum-mechanical violation of which implies the nonexistence of the non-contextual hidden variable models. Several important applications of contextuality have recently been found in the certification of random number \cite{deng} as well as the speeding up of quantum algorithms \cite{Howard}.

Graph theory has had wide applications in information theory. Very recently, Cabello, Severini and Winter (CSW) have introduced a general graph-theoretic approach for studying contextuality \cite{CSW14}, this allows to show that quantum contextuality is closely related to the Lov\'asz number \cite{Lovasz79}, an important parameter used in optimization and information theory. For a given noncontextuality inequality
\begin{eqnarray}\label{inequ1}
S = \sum_i w_i\langle P_i \rangle \leq \alpha
\end{eqnarray}
with some exclusivity relation, it can always be represented by an exclusivity  graph. The concrete way is as follows. There is an exclusivity relation between two rank-$1$ projective measurements $P_i$ and $P_j$ if they are orthogonal to each other. There is an edge $e_{ij}$ between two vertices $i,j$ if and only if $P_i,P_j$ are exclusive. The set of all the exclusivity relations of all $P_i$'s is said to be the exclusivity relation of the  noncontextual inequality. In addition, one can associate each $P_i$ with a weight $w_i$ to the $i$-th vertex. In this way, the constructed graph $G=(V,E,W)$ is called as the exclusivity graph of the noncontextual inequality, where $V$ is the set of vertices, $E$ is the set of edges and $W$ is the set of weights. And as shown in \cite{CSW14}, the classical bound of $S$ is just the independence number $\alpha$ of the graph $G$.

It is interesting to study the maximal quantum violation $S^{\rm max}$ for the inequality (\ref{inequ1}) as well as the optimal representation for projectors $P_i$'s. Remarkably Ref. \cite{CSW14} has pointed out that $S^{\rm max}$ is nothing but the Lov\'asz number $\vartheta$ of the graph $G$ \cite{Lovasz79}. Correspondingly, the optimal representation for the projectors is often called as the Lov\'asz-optimum orthogonal representation (LOOR). Let us take the  Klyachko-Can-Binicio\u{g}lu-Shumovsky (KCBS) inequality \cite{kcbs} as an example. The KCBS inequality is the simplest noncontextual inequality for the three-dimensional system, in the sense that it requires the minimal number of projective measurements.

The KCBS inequality is given by
\begin{equation}
S_{\rm KCBS} = \sum_{j=1}^5 \langle P_j \rangle \leq 2,
\end{equation}
where $\langle P_j \rangle= {\rm Tr}[\rho P_j]$, $\rho$ is quantum state, and $P_j$'s
are rank-1 projective measurements with exclusivity relations: $P_jP_{j+1}=0$ ($j=1,2,3,4$) and $P_5P_1=0$. The exclusivity graph for the KCBS inequality is a pentagon graph, in which the weights of all vertices are equal to 1 (see Fig. \ref{fig1}).

\begin{figure}[!h]
\includegraphics[width=5cm]{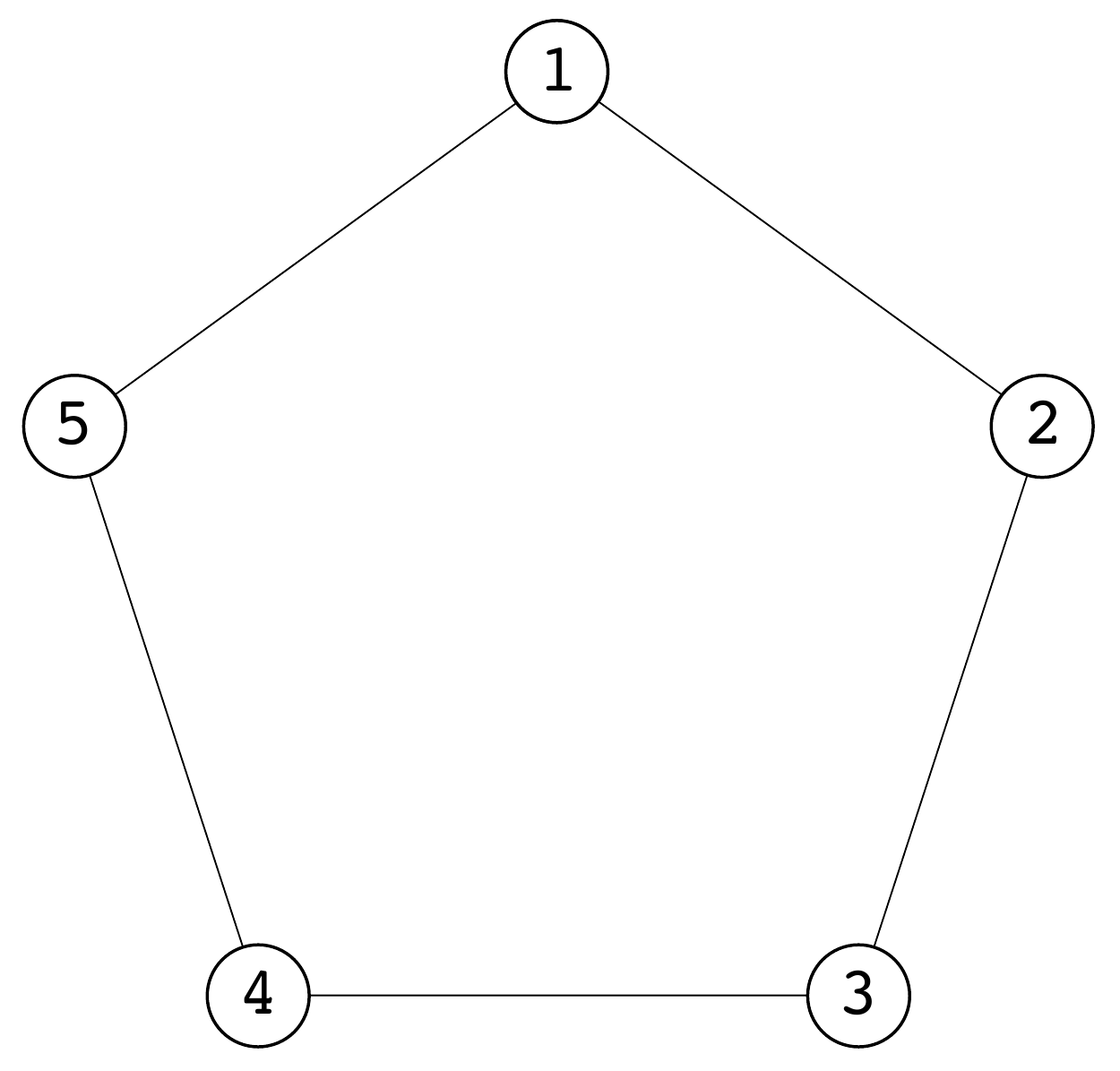}
\caption{The exclusivity graph of the KCBS inequality. }\label{fig1}
\end{figure}
It has been known that the maximal quantum violation is $S_{\rm KCBS}^{\rm max}=\vartheta=\sqrt{5}$, and the corresponding Lov\'asz-optimum orthogonal representation can be realized in the three-dimensional real Hilbert space as

\begin{eqnarray}\label{kcbslovasz1}
\rho&=&|\psi\rangle\langle\psi|,\nonumber\\
P_j&=& |v_j\rangle\langle v_j|, \;\;\; (j=1,2,3,4),
\end{eqnarray}
where
\begin{eqnarray}\label{kcbslovasz}
|\psi\rangle&=& (1,0,0)^T, \nonumber\\
|v_j\rangle&=&\tau\; (1,0,0)^T+\sqrt{1-\tau^2}\; (0,\cos \varphi_j,\sin \varphi_j)^T,\nonumber\\
\tau &=&\sqrt{\frac{1}{\sqrt{5}}},\;\; \varphi_j= \frac{2\pi}{5} (2 j-1),
\end{eqnarray}
and ``$T$" means transpose of a matrix.

For a general exclusivity graph of noncontextual inequality, the original Lov\'asz number is defined as
\begin{equation}
\vartheta(G) = \max \sum_{j\in V} w_j (\langle \psi|v_j\rangle)^2,
\end{equation}
where $|\psi\rangle$ and $|v_j\rangle, j\in V$ run over all possible \emph{real} unit vectors such that $\langle v_i|v_j\rangle=0$ if vertices $i,j$ are connected. Note that  in quantum theory (QT) $S^{\rm max}$ can be always obtained for a pure state $|\psi\rangle$ since the noncontextual inequality is linear, and each $P_j$ can be written as a form of $|v_j\rangle\langle v_j|$. Then the definition of Lov\'asz number directly implies that the upper bound of $S$ in quantum case is just $\vartheta(G)$ if $S^{\rm max}$ is always obtained for a set of projectors $P_j$ in a real Hilbert space of suitable dimension.

Quantum theory and quantum information theory are based on the Hilbert space. It has been claimed that ``\textit{...Taken into account that the maximum value of $S$ in QT is always obtained for a quantum pure state $|\psi\rangle$ and a set of projectors $\Pi_i$ in a real Hilbert space of suitable dimension}" \cite{CSW14}, namely, the Lov\'asz-optimum orthogonal representation can be always realized in the real Hilbert space. To our knowledge, a detailed proof for the claim has not been given in the literature. The purpose of this paper to provide such a proof. 

 The paper is organized as follows. In Sec. II,  we study the exclusivity graph originated from the work of Bengtssona, Blanchfielda, and Cabello (BBC), the BBC-21-Ray \cite{bbc21ray}, in which the vectors $|v_j\rangle$'s have been given by the complex unit vectors. This in turn provides the first example for a nontrivial realization of the LOOR in the \emph{complex} Hilbert space, and also arises immediately a natural question: for BBC-21-Ray, can one realize its LOOR in the real Hilbert space of suitable dimension (as the claim mentioned above)? The answer is positive. In Sec. III, we discuss the relation between $\vartheta_c(G)$ and $\vartheta(G)$, which are the Lov\'asz numbers in terms of complex and real unit vectors, respectively. In Sec. IV and Sec. V, we give two procedures to construct a real LOOR from a complex one. The first procedure is in the operator perspective while the second one is in the vector perspective. Conclusion is made in the last section.


\section{BBC-21-ray and its LOOR }
The noncontextual inequality of BBC-21-Ray is given by  \cite{bbc21ray}
\begin{eqnarray}\label{bbcsic}
S_{\rm BBC} = 3\sum_{j=1}^9 \langle P_j \rangle + 5\sum_{j=10}^{21} \langle P_j\rangle \leq 27
\end{eqnarray}
where $P_j$'s satisfy the exclusivity relation as shown in its exclusivity graph (see Fig.\ref{bbc21rayfig}). The 21 complex unit vectors $|v_j\rangle$'s are as follows:
\begin{eqnarray}\label{vbbc}
&\frac{1}{\sqrt{2}}(0,1,-1)^T,\frac{1}{\sqrt{2}}(1,0,-1)^T,\frac{1}{\sqrt{2}}(1,-1,0)^T,\nonumber\\
&\frac{1}{\sqrt{2}}(0,1,-e^{-\frac{2 i \pi}{3}})^T,\frac{1}{\sqrt{2}}(1,0,-e^{-\frac{2 i \pi}{3}})^T,\frac{1}{\sqrt{2}}(1,-e^{-\frac{2 i \pi}{3}},0)^T,\nonumber\\
&\frac{1}{\sqrt{2}}(0,1,-e^{\frac{2 i \pi }{3}})^T,\frac{1}{\sqrt{2}}(1,0,-e^{\frac{2 i \pi }{3}})^T,\frac{1}{\sqrt{2}}(1,-e^{\frac{2 i \pi }{3}},0)^T,\nonumber\\
&(1,0,0)^T,(0,1,0)^T,(0,0,1)^T,(1,1,1)^T/\sqrt{3},\nonumber\\
&\frac{1}{\sqrt{3}}(1,1,e^{\frac{2 i \pi }{3}})^T,\frac{1}{\sqrt{3}}(1,1,e^{-\frac{2 i \pi}{3}})^T,\frac{1}{\sqrt{3}}(1,e^{-\frac{2 i \pi}{3}},1)^T,\nonumber\\
&\frac{1}{\sqrt{3}}(1,e^{-\frac{2 i \pi}{3}},e^{\frac{2 i \pi }{3}})^T,\frac{1}{\sqrt{3}}(1,e^{-\frac{2 i \pi}{3}},e^{-\frac{2 i \pi}{3}})^T,\frac{1}{\sqrt{3}}(1,e^{\frac{2 i \pi }{3}},1)^T,\nonumber\\
&\frac{1}{\sqrt{3}}(1,e^{\frac{2 i \pi }{3}},e^{\frac{2 i \pi }{3}})^T,\frac{1}{\sqrt{3}}(1,e^{\frac{2 i \pi }{3}},e^{-\frac{2 i \pi}{3}})^T.
\end{eqnarray}

\begin{figure}
\centering
\includegraphics[width=7cm]{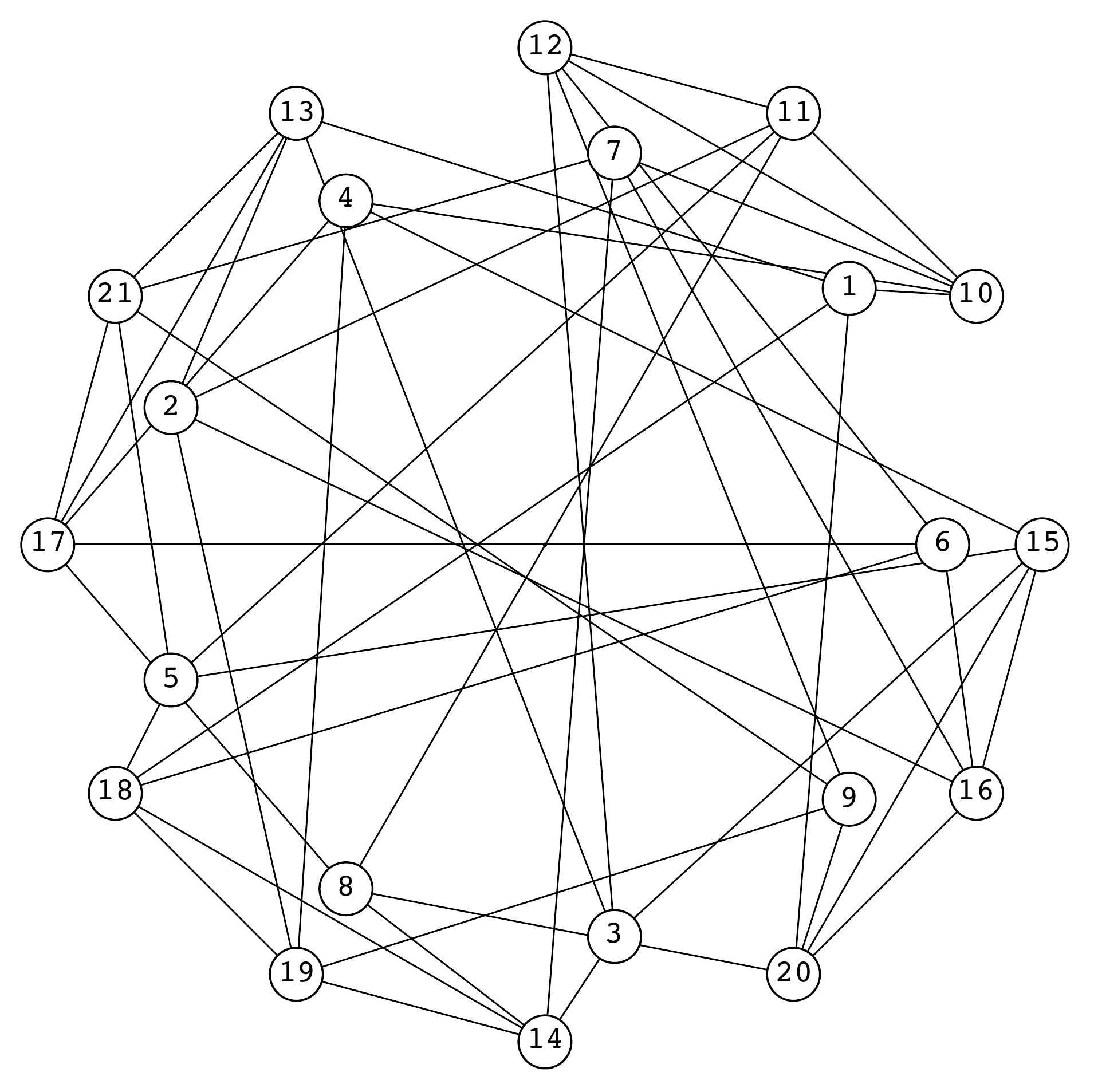}
\caption{The exclusivity graph of BBC-21-ray.}
\label{bbc21rayfig}
\end{figure}

Originally, BBC-21-Ray was proposed to develop a state-independent noncontextuality (SIC) inequality for a qutrit (a three-dimensional quantum system).
This inequality holds for any noncontextual theory, while the upper bound of $S$ is 
\begin{eqnarray}\label{bbclovasz}
S_{\rm BBC}^{\rm max}= 29
\end{eqnarray}
holds for any state of a qutrit. In addition, the  Lov\'asz number of the graph can be computed as $\vartheta(G) =29$ by using the semidefinite programming, thus $S_{\rm BBC}^{\rm max}=\vartheta(G)= 29$.

Let us generalize a little bit the Lov\'asz number in terms of complex unit vectors as: 

\begin{equation}
\vartheta_c(G) = \max \sum_{j\in V} w_j |\langle \psi|v_j\rangle|^2,
\end{equation}
where $|\psi\rangle$ and $|v_j\rangle, j\in V$ run over all possible \emph{complex} unit vectors such that $\langle v_i|v_j\rangle=0$ if vertices $i, j$ are connected. By choosing 
$\rho = |\psi\rangle\langle\psi|$, $|\psi\rangle =(1,0,0)^T$, and the complex vectors $|v_j\rangle$'s in Eq. (\ref{vbbc}), one immediately has 
\begin{equation}
S_{\rm BBC}^{\rm max}=\vartheta_c(G)= \vartheta(G)= 29,
\end{equation}
and automatically obtains a complex realization of LOOR in the Hilbert space.

Here, we would like to address that, the BBC-21-Ray graph indeed can have a real realization of LOOR in the Hilbert space by performing the procedures in Sec. III or Sec. IV. 
The real LOOR can be realized in $2\times 3-1=5$ dimension. The result is $|\psi\rangle = (1,0,0,0,0)^T$ and $|v_j\rangle$'s are as follows

\begin{eqnarray}\label{21real}
&\left(0,\frac{1}{\sqrt{2}},\frac{-1}{\sqrt{2}},0,0\right)^T,\left(\frac{1}{\sqrt{2}},0,\frac{-1}{\sqrt{2}},0,0\right)^T,\left(\frac{1}{\sqrt{2}},\frac{-1}{\sqrt{2}},0,0,0\right)^T,\nonumber\\
&\left(0,\frac{1}{\sqrt{2}},\frac{1}{2 \sqrt{2}},0,\frac{\sqrt{3}}{2\sqrt{2}}\right)^T,\left(\frac{1}{\sqrt{2}},0,\frac{1}{2 \sqrt{2}},0,\frac{\sqrt{3}}{2\sqrt{2}}\right)^T,\nonumber\\
&\left(\frac{1}{\sqrt{2}},\frac{1}{2 \sqrt{2}},0,\frac{\sqrt{3}}{2\sqrt{2}},0\right)^T,\left(0,\frac{1}{\sqrt{2}},\frac{1}{2 \sqrt{2}},0,\frac{-\sqrt{3}}{2\sqrt{2}}\right)^T,\nonumber\\
&\left(\frac{1}{\sqrt{2}},0,\frac{1}{2 \sqrt{2}},0,\frac{-\sqrt{3}}{2\sqrt{2}}\right)^T,\left(\frac{1}{\sqrt{2}},\frac{1}{2 \sqrt{2}},0,\frac{-\sqrt{3}}{2\sqrt{2}},0\right)^T,\nonumber\\
&(1,0,0,0,0)^T,(0,1,0,0,0)^T,(0,0,1,0,0)^T,\nonumber\\
&\left(\frac{1}{\sqrt{3}},\frac{1}{\sqrt{3}},\frac{1}{\sqrt{3}},0,0\right)^T,\left(\frac{1}{\sqrt{3}},\frac{1}{\sqrt{3}},\frac{-1}{2 \sqrt{3}},0,\frac{1}{2}\right)^T,\nonumber\\
&\left(\frac{1}{\sqrt{3}},\frac{1}{\sqrt{3}},\frac{-1}{2 \sqrt{3}},0,\frac{-1}{2}\right)^T,\left(\frac{1}{\sqrt{3}},\frac{-1}{2 \sqrt{3}},\frac{1}{\sqrt{3}},\frac{-1}{2},0\right)^T,\nonumber\\
&\left(\frac{1}{\sqrt{3}},\frac{-1}{2 \sqrt{3}},\frac{-1}{2 \sqrt{3}},\frac{-1}{2},\frac{1}{2}\right)^T,\left(\frac{1}{\sqrt{3}},\frac{-1}{2 \sqrt{3}},\frac{-1}{2 \sqrt{3}},\frac{-1}{2},\frac{-1}{2}\right)^T,\nonumber\\
&\left(\frac{1}{\sqrt{3}},\frac{-1}{2 \sqrt{3}},\frac{1}{\sqrt{3}},\frac{1}{2},0\right)^T,\left(\frac{1}{\sqrt{3}},\frac{-1}{2 \sqrt{3}},\frac{-1}{2 \sqrt{3}},\frac{1}{2},\frac{1}{2}\right)^T,\nonumber\\
&\left(\frac{1}{\sqrt{3}},\frac{-1}{2 \sqrt{3}},\frac{-1}{2 \sqrt{3}},\frac{1}{2},\frac{-1}{2}\right)^T.
\end{eqnarray}
Then
\begin{eqnarray}\label{matrix}
&&3\sum_{j=1}^9 |v_j\rangle\langle v_j| + 5\sum_{j=10}^{21} |v_j\rangle\langle v_j|\nonumber\\
&&=
\begin{bmatrix}
 29 & 0 & 0 & 0 & 0 \\
 0 & \frac{77}{4} & 0 & 0 & 0 \\
 0 & 0 & 17 & 0 & 0 \\
 0 & 0 & 0 & \frac{39}{4} & 0 \\
 0 & 0 & 0 & 0 & 12 \\
\end{bmatrix},
\end{eqnarray}
this yields directly $\vartheta(G)= 29$.

\emph{Remark 1.}--- The 21 complex unit vectors $|v_j\rangle$'s in Eq. (\ref{vbbc}) cannot be rotated to 21 real unit vectors simultaneously by a general unitary transformation. This implies that the realization of the complex vectors together with the complex LOOR are nontrivial.

\emph{Remark 2.}--- The maximal eigenvalue of the matrix in the right-hand side of Eq. (\ref{matrix}) is still $29$ implies the set of $|\psi\rangle$ and $|v_j\rangle$'s is indeed a LOOR in the real Hilbert space. However, the other eigenvalues are less than $29$ implies that the inequality (\ref{bbcsic}) is no longer an SIC inequality in $5$-dimensional Hilbert space. Thus, the complex Hilbert space is still needed if we want to keep some special properties of the inequality, like state-independent noncontextuality \cite{Yu-Oh,Cabello5,Cabello1,XCS}.

\section{$\vartheta_c(G)$ versus $\vartheta(G)$}

For convenience, let's firstly list some symbol assumptions:

\begin{enumerate}
\item $x^r,x^i$ is the real part and the imaginary part of $x = x^r + i x^i$ respectively for $x$ is a number, a vector or a matrix.
\item $A\cdot B$ means $\tr A^HB$ where $A,B$ are two matrices.
\item $A\succeq 0$ means that $A$ is a positive semidefinite matrix.
\item $\mathcal{R}_{n\times n}$ is the set of real $n\times n$ dimensional matrices.
\item $\mathcal{C}_{n\times n}$ is the set of complex $n\times n$ dimensional matrices.
\end{enumerate}

To prove that upper bound of $S$ in quantum case is the Lov\'asz number is equivalent to prove 
\begin{eqnarray}
\vartheta_c(G) = \vartheta(G).
\end{eqnarray}
It's obvious that  $\vartheta_c(G) \geq \vartheta(G)$, we also need to show that $\vartheta_c(G) \leq \vartheta(G)$ for any graph $G$.

It has been prove that  \cite{Lovasz79}
\begin{eqnarray}
\vartheta(G) &=& \max_X J\cdot X\nonumber\\
&\text{s.t.}&I \cdot X = 1,\nonumber\\
&&J_{ij} \cdot X = 0, \;\; (i, j) \in E,\nonumber\\
&&X \succeq 0,~X\in \mathcal{R}_{n\times n},
\end{eqnarray}
where $J$ is a matrix full of $1$ while $J_{ij}$ is the matrix whose $i$-rule $j$-column element is $1$ and the rest are $0$.

It can be proven that $\vartheta_c(G) \leq \vartheta'_c(G)$, where
\begin{eqnarray}\label{complexsdp}
\vartheta'_c(G) &=& \max_X J\cdot X = \max_X J\cdot X^r\nonumber\\
&\text{s.t.}&I \cdot X = 1,\nonumber\\
&&J_{ij} \cdot X = 0, \;\; (i, j) \in E,\nonumber\\
&&X \succeq 0,~X\in \mathcal{C}_{n\times n}.
\end{eqnarray}
The proof is as follows. If the set of complex unit vectors $|d\rangle$ and $|v_i\rangle$'s is an optimal solution, that is, $\vartheta_c(G) = \sum_{i=1}^n |\langle d|v_i\rangle|^2$. Then it's directly to see that the set of complex unit vectors $U|d\rangle$ and $U|v_i\rangle$'s is also optimal. Without loss of generality, we can always assume that $|d\rangle = (1,\ldots,1)/\sqrt{n}$. Denote $X$ as the matrix such that $x_{ij} = \langle v_i|v_j\rangle/n$. One can check that $X$ satisfies the conditions in Eq. \eqref{complexsdp}, and $J\cdot X = \sum_{i=1}^n|\langle d|v_i\rangle|^2 = \vartheta_c(G) \leq \vartheta'_c(G)$.

Define
\begin{eqnarray}\label{complexsdp2}
\vartheta''_c(G) &=& \max_X J\cdot X^r\nonumber\\
&\text{s.t.}&I \cdot X^r = 1,\nonumber\\
&&J_{ij} \cdot X^r = 0, \;\; (i, j) \in E,\nonumber\\
&&X^r \succeq 0,~X\in \mathcal{C}_{n\times n}.
\end{eqnarray}
Since $\max_X J\cdot X = \max_X J\cdot X^r\nonumber$, $I \cdot X = I \cdot X^r$, $J_{ij} \cdot X = 0$ implies $J_{ij} \cdot X^r = 0$ and $X^r \succeq 0$ is necessary for $X \succeq 0$, so $\vartheta'_c(G) \leq \vartheta''_c(G)$. On the other hand, it's directly to see that $\vartheta''_c(G) \leq \vartheta(G)$ because $X^r \succeq 0,~X^r\in \mathcal{R}_{n\times n}$ is necessary for $X^r \succeq 0,~X\in \mathcal{C}_{n\times n}$.

After all, $\vartheta_c(G) \leq \vartheta(G)$. Thus $\vartheta_c(G) = \vartheta(G)$, which implies the quantum upper bound of the noncontextual inequality equals to the original Lov\'asz number.

\emph{Remark 3.}--- For the general noncontextual inequality (\ref{inequ1}), whose weights maybe not are all $1$'s, then its exclusivity graph is the weighted graph. We can change the maximized object $J\cdot X$ to $W\cdot X$, where $w_{ij} = \sqrt{w_iw_j}$. And the rest of the proof is same.
On the other hand, we can give two procedures to construct a real optimal solution (the LOOR realized in the real Hilbert space) from a complex one, which shows $\vartheta_c(G) = \vartheta(G)$ from another point of views. One procedure is in the operator perspective while the other one is in the vector perspective.

\section{Projector construction procedure}

Assume the set of pure state $\rho = |\psi\rangle\langle \psi|$ and rank-$1$ projectors $P_i = |v_i\rangle\langle v_i|$'s is an optimal complex solution for the inequality $S = \sum_i w_i\langle P_i\rangle \le \alpha$, that is, $\vartheta = \sum_i w_i {\rm Tr}[P_i\cdot\rho]$. Then the first step to construct a real optimal solution is constructing real projector $Q_i$'s such that
\begin{eqnarray}
Q_i = \begin{bmatrix}P_i^r&&-P_i^i\\P_i^i&&P_i^r\end{bmatrix}.
\end{eqnarray}
It's easy to find that $Q_i$'s are rank-$2$ projectors and $Q_i\cdot Q_j=0$ if $P_i\cdot P_j=0$. What's more, we claim that the maximal eigenvalue of
$\sum_i w_i Q_i$
is the same as $\sum_i w_iP_i$, that is, the Lov\'asz number. Let's denote 

\begin{eqnarray}
\lambda I_n-\sum_i w_iP_i = A+iB,
\end{eqnarray}
 then 
 \begin{eqnarray}
 \lambda I_{2n} - \sum_i w_i Q_i = \begin{bmatrix}A&-B\\B&A\end{bmatrix}.
 \end{eqnarray}
  The claim is implied by the following lemma.
\begin{lemma}
$A,B$ are two $n\times n$ dimensional real matrice. Then $A+iB \succeq 0$  if and only if $\begin{bmatrix}A&-B\\B&A\end{bmatrix} \succeq 0$.
\end{lemma}
\begin{proof}
Denote $\mathcal{V} = \begin{bmatrix}1&i\\-i&1\end{bmatrix}\otimes I_n$, then $\mathcal{V}$ is a positive definite Hermite matrix. So,
\begin{eqnarray}
\begin{bmatrix}A&-B\\B&A\end{bmatrix} \succeq 0 &\Leftrightarrow& \mathcal{V} \begin{bmatrix}A&-B\\B&A\end{bmatrix} \mathcal{V}= 2\mathcal{V}\otimes (A+iB)\succeq 0 \nonumber\\
&\Leftrightarrow& A+iB \succeq 0.
\end{eqnarray}
\end{proof}

Till now, we have constructed the real rank-$2$ projectors from the complex ones. We continue to construct the real rank-$1$ projectors based on the rank-$2$ ones in the next step.

Let's denote 
\begin{eqnarray}
\tilde{\rho} = \frac{1}{2}\begin{bmatrix}\rho^r&-\rho^i\\\rho^i&\rho^r\end{bmatrix},
\end{eqnarray}
 directly calculation shows that
\begin{eqnarray}
\tilde{\rho}\cdot\sum_i w_i Q_i = \rho\cdot\sum_i w_i P_i = \vartheta,
\end{eqnarray}
which means the rank-$2$ mixed state $\tilde{\rho}$ is optimal. If we choose a rank-$1$ decomposition of $\tilde{\rho}=\tilde{\rho}^1+\tilde{\rho}^2$, then $\tilde{\rho}^1, \tilde{\rho}^2$ are all optimal. And we can always decompose $Q_i$ into rank-$1$ projectors $Q_i^1, Q_i^2$ such that $\tilde{\rho}_1\cdot Q_i^2=0$. Then
\begin{eqnarray}
\tilde{\rho}^1\cdot \sum_i w_i Q_i^1 = \tilde{\rho}^1\cdot \sum_i w_i Q_i = \vartheta,
\end{eqnarray}
and $Q_i^1\cdot Q_j^1=0$ if $Q_i\cdot Q_j=0$. It's easy to find that $Q_i^1 = Q_i\tilde{\rho}^1Q_i/(\tilde{\rho}^1\cdot Q_i), Q_i^2 = Q_i - Q_i^1$ is such a construction.

After all, the set of $\tilde{\rho}^1$ and $Q_i^1$'s is the real rank-$1$ optimal solution for the quantum case.

\section{Vector construction procedure}

We go on to give the vector construction procedure, which is essentially equivalent to the projector one.

Assume the set of $d$ dimensional unit vectors $|\psi\rangle$ and $|v_i\rangle$'s is an optimal compelx solution, that is, $\vartheta = \sum_i w_i |\langle \psi|v_i\rangle|^2$. Then the construction procedure is as following.

The first step is to find a basis $|1\rangle, |2\rangle, \ldots, |d\rangle$ such that $|\psi\rangle = |1\rangle$.
The second step is to construct $|u_i\rangle$ as $|u_i\rangle = e^{-i \alpha_i} |v_i\rangle$, where $\alpha_i$ is such an angle that $\langle \psi|v_i\rangle = e^{i\alpha_i} |\langle \psi|v_i\rangle|$. Since $\langle \psi|u_i\rangle = \langle \psi|u_i^r\rangle + i\langle \psi|u_i^i\rangle \geq 0$, $\langle \psi|u_i^i\rangle$ will always be $0$.

The third step is mapping any $|v\rangle = \sum_{i=1}^d c_i |i\rangle$ to the $2d$-dimensional vector 
\begin{eqnarray}
\mathcal{M}(|v\rangle) = \sum_{i=1}^d [c_i^r|i\rangle|1\rangle + c_i^i|i\rangle|2\rangle].
 \end{eqnarray}
 
 Denote $|\phi\rangle = \mathcal{M}(|\psi\rangle)$, $|\omega_i\rangle = \mathcal{M}(|v_i\rangle)$ for $i\in V$, and $S^r = \sum_i w_i (\langle \phi|\omega_i\rangle)^2$, then we can check that
\begin{eqnarray}
S^r &=& \sum_i w_i(\langle \psi|u^r_i\rangle)^2
= \sum_i w_i(\langle \psi|u_i\rangle)^2\nonumber\\
&=& \sum_i w_i|\langle \psi| e^{-i \alpha_i} |v_i\rangle|^2\nonumber\\
&=&  \sum_i w_i|\langle \psi|v_i\rangle|^2 = \vartheta.
\end{eqnarray}
And the exclusivity relations
\begin{eqnarray}
\langle \omega_i|\omega_j\rangle = \langle u^r_i|u^r_j\rangle + \langle u^i_i|u^i_j\rangle = (\langle u_i|u_j\rangle)^r = 0.
\end{eqnarray}
Denote $|\varphi\rangle = |\psi\rangle|2\rangle$, then $\langle \phi|\varphi\rangle = 0$ and
\begin{eqnarray}
\langle \varphi|\omega_i\rangle = \langle \psi|u^i_i\rangle = 0,
\end{eqnarray}
which means $|\omega_i\rangle$'s and $|\psi\rangle$ are in a $(2d-1)$-dimensional subspace.
Thus, we can reduce the $2d$-dimensional unit vectors $|\phi\rangle$ and $|\omega_i\rangle$'s to $(2d-1)$-dimensional ones without changing the exclusivity relations and $S^r = \vartheta(G)$.

Till now, we have completed the vector procedure and the real unit vectors for the BBC-21-Ray as shown in Eq. (\ref{21real}) have been successfully constructed by these two procedures.

\section{Conclusion}

In conclusion, we have proved in detail that the Lov\'asz-optimum orthogonal representation for any exclusivity graph can be realized in the real Hilbert space of suitable dimension. Explicitly, if there is a Lov\'asz-optimum orthogonal representation realized in the $d$-dimensional complex Hilbert space, then there always exists a corresponding Lov\'asz-optimum orthogonal representation in the $(2d-1)$-dimensional real Hilbert space. Very recently, a general set of SIC has been developed for a single qutrit \cite{XCS}, in which there are $(3+3k+k^2)$ complex rays that involve the BBC-21-Ray as a special case. The real as well as the complex Lov\'asz-optimum orthogonal representation could be obtained accordingly, which we shall investigate subsequently.

\begin{acknowledgments}
J.L.C. is supported by the National Basic Research Program (973 Program) of China under Grant No. 2012CB921900 and the NSF of China (Grant Nos. 11175089 and 11475089). This work is also partly supported by the National Research Foundation and the Ministry of Education, Singapore.
\end{acknowledgments}

\end{document}